\documentclass[conference,letterpaper]{IEEEtran}
\IEEEoverridecommandlockouts
\usepackage[utf8]{inputenc} 
\usepackage[T1]{fontenc}
\usepackage{url}
\usepackage{ifthen}
\usepackage{cite,dsfont,proba}
\usepackage[cmex10]{amsmath}

\usepackage{amsmath,amssymb,amsthm}
\usepackage{tikz,enumerate}
\theoremstyle{plain}

\newtheorem{theorem}{Theorem} 
\newtheorem{proposition}{Proposition}

\newtheorem{lemma}{Lemma}

\theoremstyle{remark}

\newtheorem{remark}{Remark}

\theoremstyle{definition}
\newtheorem{definition}{Definition}

\usepackage{tikz}

\newtheorem{example}{Example}
\begin{document}

\title{Subjective Distortion: Achievability and Outer Bounds for Distortion Functions with Memory} 

\author{%
  \IEEEauthorblockN{Hamidreza Abin, Amin Gohari}
  \IEEEauthorblockA{Dept. of Information Engineering\\
                    The Chinese University of Hong Kong\\
                    Sha Tin, NT, Hong Kong}
  \and
  \IEEEauthorblockN{Andrew W.~Eckford}
  \IEEEauthorblockA{Dept.~of Electrical Engineering and Computer Science\\ 
                    York University\\
                    Toronto, Canada}%
\thanks{This work was supported in part by the John Templeton Foundation as Project
\#63644.}%
}

\maketitle

\begin{abstract}
In some rate-distortion-type problems, the required fidelity of information is affected by past actions. As a result, the distortion function depends not only on the instantaneous distortion between a source symbol and its representation symbol, but also on past representations. In this paper, we give a formal definition of this problem and introduce both inner (achievable) and outer bounds on the rate-distortion tradeoff. We also discuss convexification of the problem, which makes it easier to find bounds. Problems of this type  arise in biological information processing, as well as in recommendation engines; we provide an example applied to a simplified biological information processing problem.

\end{abstract}


\section{Introduction}


Consider a viewer of a video streaming platform. The viewer has many options in terms of content {\em a priori}, but once they select the first episode of a series, the other episodes in that series become much more desirable to them. Alternatively, consider an animal that needs two nutrients to survive, food and water. Searching for food might satiate the animal's hunger, but lead to increasing thirst, changing the relative reward corresponding to food and water. In both of the above examples, the utility or desired fidelity for certain kinds of information is dependent on the user's own past information-gathering choices. Given this dependence, what is an optimal information-processing strategy?

These problems are examples of {\em subjective information} \cite{barker2022subjective}. In a subjective information problem, the semantics or value of information is significant, as in semantic or task-oriented communication \cite{gunduz2022beyond}, but these semantics are modified by the user's internal state, which in turn is affected by the user's past information-gathering or information-processing choices. Subjective information problems have been analyzed in terms of biological resource gathering, such as in cases where an organism has finite sensing capabilities but must split them to seek two different resources \cite{barker2023metric,barker2024fitness}. Time- and user-dependent preferences have also been studied in recommendation engines as the ``taste-distortion'' problem \cite{zhao2021rabbit}. Analogous problems occur in channel coding, e.g. \cite{koch2009channels}.


Rate-distortion theory \cite{berger1971,berger2002living} and semantic information theory \cite{bartlett2025physics} have been proposed as methodologies for biological communication problems, and we are particularly interested in these applications. 
There has been much recent work in the {\em fitness value of information} \cite{donaldson2010fitness,rivoire2011value}, applying rate-distortion and similar methodologies to explore optimal information processing and growth tradeoffs in evolutionary systems. The result that single-letter codes can achieve optimal rate-distortion tradeoffs \cite{gastpar2003code}, particularly in systems that optimize growth rate (a commonly used measure for evolutionary fitness \cite{wu2013interpretations}), implies that these tradeoffs may be achievable by organisms \cite{moffett2025kelly}. 
Semantic communication, meanwhile, has been studied in a rate-distortion context in conventional communication systems \cite{chai2023rate,stavrou2022rate,liu2021rate}.

While growth-rate analysis has received considerable attention in the biophysics literature, novel methodologies, such as subjective information, are required to describe complicated, time-varying biological phenomena beyond toy problems. In this paper, we generalize the rate-distortion framework to the subjective case where the distortion function is dependent on the most recent encoded symbol. Our main contribution is an achievability result and an outer bound, placing limits on the exact form of the rate-distortion tradeoff in these problems. We also give some brief examples illustrating how these results might be used in practical situations.

The remainder of the paper is organized as follows. In Section II, we introduce the rate-distortion problem setup. In Section III, we give our achievability results, and discuss the convexity of the solution space. In Section IV, we introduce outer bounds on the rate-distortion tradeoff. In Section V, we give a numerical example, and conclude in Section VI.

\section{Problem Setup and Scientific Motivation}
\label{sec:setup}

We consider a problem of the rate-distortion type. Consider an environment characterized by a state $X_i$ at each discrete time step $i$. We assume that the sequence of states $\{X_i\}_{i \geq 1}$ comprises independent and identically distributed (i.i.d.) random variables governed by the probability distribution $p_X(x)$. An agent operating within this environment takes an action $Y_i$ at time $i$. We assume that the agent's reward at any given time depends on both the current environmental state and the agent's action history. While this work focuses specifically on the dependence between the reward and the current and immediately preceding actions, our framework is readily generalizable to scenarios involving longer memory horizons.

Furthermore, rather than maximizing a reward function $r(x_i, y_i, y_{i-1})$, we formulate the problem as the minimization of an \emph{opportunity cost} $d(x_i, y_i, y_{i-1})$, defined as $d(x_i, y_i, y_{i-1}) = r_{\max} - r(x_i, y_i, y_{i-1})$, where $r_{\max}$ represents the maximum achievable reward. Assuming that the agent is active within a window of $n$ time steps, it seeks to minimize the expected cost:
\begin{equation}
\frac1n\sum_{i=1}^n\mathbb{E}[d(X_i, Y_i, Y_{i-1})],\label{eqnRD}
\end{equation}
where $Y_0$ is a fixed initial symbol.
We assume that the agent possesses prior side information about the environment. Moreover, at time $i$, the agent may gather further information through sensing or observing its reward/opportunity cost at earlier times $j<i$. Let $M$ denote the total collective information that the agent can obtain about the sequence $X^n$ (comprising initial information and observations during times $1\leq i\leq n$). We assume that the agent is information-limited; specifically, $I(M;X^n)\leq nR$, where $R$ denotes an upper limit on the rate of information transfer to the agent.

To study the fundamental limit of the impact of $R$ on the agent's distortion, we adopt an idealized assumption. We assume that the agent can query a genie for \emph{any} side information $M$ it desires (subject to the constraint $I(M;X^n)\leq nR$). The genie provides this information at the very beginning, allowing the agent to utilize it to generate the action sequence $(Y_1, Y_2, \dots, Y_n)$. 
Thus, while in practice the agent may gather information during time instances $1\leq i\leq n$, to upper-bound performance, we instead assume a genie provides all allowed information at time $0$ and no further observations occur. This idealized setting provides an upper bound on the performance of the real system, as it grants the agent significant freedom to select the type of information it requires. For instance, the agent could utilize its $nR$ bits of information quota to query the genie about states $X_i$ for $i<t$ (for some time $t$), while disregarding $X_i$ for $i>t$.

The described setting aligns naturally with the classical rate-distortion framework in information theory. Consider $X^n$ as the source and $M$ as the ``message'' transmitted to a decoder, which subsequently generates the output sequence $Y^n$. The objective is to minimize the total distortion, as defined in \eqref{eqnRD}. While classical theory typically assumes that $M \in \{0,1\}^{nR}$, the fundamental results remain valid under the relaxed constraint $I(M;X^n) \leq nR$.\footnote{To demonstrate this formally, we note that the channel $p_{M|X^n}$ can be synthesized using $I(M;X^n)$ bits of information in the presence of shared randomness. Furthermore, the lossy compression rate-distortion tradeoff remains unchanged whether shared randomness is available or not~\cite{cuff2013distributed}.}

In terms of physical systems that operate on this principle, we are motivated by encoders that physically change in order to generate output symbols, where the change leads to a cost; an example is given in Section \ref{sec:Examples}.
For simplicity, we consider distortion functions with a single step of memory that depend only on the previous decoder output, and we leave extensions of this framework to future work.

\section{Rate-Distortion Setup and Achievability}

\subsection{Setting}
\label{setting:1}

We assume the source is i.i.d. according to $p_X(x)$, i.e., $p_{X^n}(x^n) = \prod_{i=1}^{n} p_X(x_i)$.  
An encoder $\mathcal{E}$ maps each sequence $x^n \in \mathcal{X}^n$ to a message $m \in \mathcal{M} = \{1,2,\dots,2^{\lfloor nR \rfloor}\}$.  
A decoder $\mathcal{D}$ maps every message $m \in \mathcal{M} = \{1,2,\dots,2^{\lfloor nR \rfloor}\}$ to a sequence $y^n \in \mathcal{Y}^n$.

We consider a bounded distortion function:
\[
d(\cdot,\cdot,\cdot): \mathcal{X} \times \mathcal{Y} \times \mathcal{Y} \to \mathbb{R}.
\]
The average distortion between $x^n \in \mathcal{X}^n$ and $y^n \in \mathcal{Y}^n$ is denoted $\bar{\mathsf{d}}(x^n, y^n)$ and is defined as
\begin{align}
    \label{eqn:average-distortion}
  \bar{\mathsf{d}}(x^n,y^n) = \frac{1}{n} \Bigl[\sum_{i=1}^{n} d(x_i, y_i, y_{i-1}) \Bigr]
\end{align}
where $y_0$ is some fixed symbol in $\mathcal{Y}$.

\begin{definition}
    \label{defn:1}
   For a given source $p_X$, rate $R$, and distortion $D$, we say that $(R,D)$ is \emph{achievable} if there exists a sequence of encoder maps $\mathcal{E}_n$ and decoder maps $\mathcal{D}_n$ such that
   \begin{align}
        \label{eqn:limsup}
       \limsup_{n \to \infty} \mathbb{E}\Bigl[ \bar{\mathsf{d}}\bigl( X^n, \mathcal{D}_n \circ \mathcal{E}_n(X^n) \bigr) \Bigr] \leq D.
   \end{align}
\end{definition}
We denote the rate-distortion function by $R(D)$, defined as the infimum of all achievable rates for a given distortion $D$.
\begin{remark} From the definition, the rate-distortion function $R(D)$ is non-increasing with respect to $D$. Furthermore, the operational definition of $R(D)$, combined with a time-sharing argument, demonstrates that $R(D)$ is a convex function of $D$. See Appendix \ref{apx:Remark-1} for details.
\label{rmk1dd}
\end{remark}
\subsection{Achievability via stationarity of a Markov process}

A multi-letter characterization of the rate-distortion region can be obtained from information spectrum methods under the assumption that the distortion function is uniformly integrable (see \cite[Theorem 5.5.1]{HanBook}). The theorem (adapted to our setup) implies that
\begin{theorem}\cite[Theorem 5.5.1]{HanBook} For a uniformly integrable distortion function, we have
$$R(D)=\inf_{\mathbf{Y}:d(\mathbf{X},\mathbf{Y})\leq D}\bar{I}(\mathbf{X};\mathbf{Y})$$
    where $\mathbf{X}=(X_1,X_2, \cdots)$ is an i.i.d.\ infinite-length source sequence, $\mathbf{Y}=(Y_1,Y_2, \cdots)$ is an arbitrary (infinite-length) reconstruction (jointly distributed with  $\mathbf{X}$),  $$d(\mathbf{X},\mathbf{Y})=\limsup_{n\rightarrow\infty}\mathbb{E}[\bar{\mathsf{d}}(X^n,Y^n)]$$
    and
    $$\bar{I}(\mathbf{X};\mathbf{Y})=
    \text{$p$-}\limsup_{n\rightarrow\infty}
    \frac1n\log\frac{p(X^n,Y^n)}{p(X^n)p(Y^n)}.$$
    \label{thmspect}
\end{theorem}

The uniformly integrable condition holds, for example, when $\mathcal{X}$ is a finite set. However, such multi-letter characterizations are generally computationally intractable. Therefore, we focus on deriving computable single-letter bounds.

\begin{theorem}\label{fthm1}
Consider a Markov chain kernel
\[
W_{X,Y|\hat{X},\hat{Y}}(x,y | \hat{x},\hat{y}) = p_X(x) \, W_{Y|X,\hat{Y}}(y | x, \hat{y})
\]
defined on the product space $\mathcal{X} \times \mathcal{Y}$. Let $\pi_{X,Y}$ denote a stationary distribution induced by this kernel. Assume that under the joint distribution
$\pi_{X,Y}(\hat{x},\hat{y}) \, W(x,y | \hat{x},\hat{y})$, the expected distortion satisfies $\mathbb{E}[d(X, Y, \hat{Y})] \leq D$.
Then, the following rate is achievable:
$$R=H(X)+H(Y|\hat Y)-H(X,Y|\hat X, \hat Y)$$
where the above is computed under $\pi_{X,Y}(\hat{x},\hat{y}) \, W(x,y | \hat{x},\hat{y})$.
\end{theorem}
\begin{proof}
    We use Theorem \ref{thmspect}. Let the sequence $(X^n, Y^n)$ be generated by the Markov chain with the given kernel $W$, initialized at the stationary distribution $\pi_{X,Y}$, i.e., $(X_1,Y_1)\sim \pi_{X,Y}$. By    
    specializing \cite[Theorem 5.5.1]{HanBook} to this Markov source, we obtain the following achievable rate:
\begin{align*}
   R &= \lim_{n \to \infty} \frac{1}{n} \, I(X^n; Y^n).
\end{align*}
Note that with the above construction of the Markov kernel, the sequence $Y^n$ will also become a Markov chain with the kernel 
\[
\sum_{x} p_X(x)\,W_{Y|X,\hat{Y}}(y|x,\hat{y}).
\]
Therefore,
\begin{align*}
   R &= \lim_{n \to \infty} \frac{1}{n} \, I(X^n; Y^n)
   \\&=H(X)+\widetilde{H}(Y) - \widetilde{H}(X,Y),
\end{align*}
where due to the Markov chain assumption
\begin{align*}\widetilde{H}(X,Y)&=H(X_2,Y_2|X_{1},Y_{1}),\\
    \widetilde{H}(Y)&=H(Y_2|Y_{1}).
\end{align*}

\end{proof}

The primary difficulty in applying the above theorem stems from the constraint 
\[
\mathbb{E}[d(X, Y, \hat{Y})] \leq D,
\]
which must be evaluated under the stationary joint distribution of \((X, Y, \hat{Y})\). This stationary distribution takes the form
\[
p_X(x) \, W(y | x, \hat{y}) \, \pi(\hat{y}),
\]
where \(\pi(\hat{y})\) denotes the stationary marginal distribution of \(\hat{Y}\).

Since \(\pi(\hat{y})\) is itself determined by the kernel \(W(y | x, \hat{y})\) through the stationary fixed-point equations, the product \(W(y | x, \hat{y}) \, \pi(\hat{y})\) becomes a nonlinear function of the kernel \(W(y | x, \hat{y})\). Consequently, the distortion constraint turns into a nonlinear condition on the choice of \(W(y | x, \hat{y})\), which complicates the optimization over achievable rates. 

Also, note that for some values of $D$, one may not be able to find a kernel \(W(y | x, \hat{y})\) for which \(
\mathbb{E}[d(X, Y, \hat{Y})] \leq D
\) holds. 

\begin{definition}
    Let $\mathtt{R}_{\mathrm{I}_1}(D)$ be the minimum rate that one can obtain from the construction in Theorem \ref{fthm1}. It should be clear from the foregoing that $R(D)\leq \mathtt{R}_{\mathrm{I}_1}(D)$. 
\end{definition}

\subsection{Achievability via memoryless kernels}

Next, we define $\mathtt{R}_{\mathrm{I}_2}(D)$ to be the minimum rate that one can obtain from the construction in Theorem \ref{fthm1} when we restrict to Markov kernels of the form
    $$W(x,y|\hat x,\hat y)=p_X(x)W(y|x).$$
    In other words, the Markov chain becomes memoryless. 

\begin{definition}
 The ``memoryless" rate-distortion function, denoted by \(\mathtt{R}_{\mathrm{I}_2}(D)\), is defined as
\begin{align}
    \mathtt{R}_{\mathrm{I}_2}(D)=\min_{W_{Y|X} \in \mathcal{W}_{D}} I(X;Y),\label{Racc:def}
\end{align}
where 
\begin{align}
    \mathcal{W}_{D}&= \Bigl\{ W_{Y|X} : \mathbb{E}_{Q_{X,Y,\hat{Y}}}[\,d(X,Y,\hat{Y})\,] \leq D,\nonumber\\&~~\quad
    Q_{X,Y,\hat{Y}}(x,y,\hat{y}) = p_X(x)W_{Y|X}(y|x)p_{Y}(\hat{y}) \Bigr\}.\label{def:W}
\end{align}
In \eqref{def:W}, distribution $p_{Y}(\hat{y})$ is defined as 
\begin{align}
    p_{Y}(\hat{y})=\sum_{x\in \mathcal{X}}p_{X}(x)W_{Y|X}(\hat{y}|x),~~\forall \hat{y}\in \mathcal{Y}.\label{def:p:Y:acc}
\end{align}
 \end{definition}
\begin{remark}
Since $\mathcal{W}_D$ is a subset of the kernels considered in Theorem \ref{fthm1}, $\mathtt{R}_{\mathrm{I}_1}(D) \leq \mathtt{R}_{\mathrm{I}_2}(D)$. However, working with $\mathtt{R}_{\mathrm{I}_2}$ is often more convenient, since the dependence of $\mathtt{R}_{\mathrm{I}_1}$ on its underlying kernel is implicit and more involved than the explicit kernel dependence of $\mathtt{R}_{\mathrm{I}_2}$.
\end{remark}

\begin{remark}
Note that the constraint in the definition of the set \(\mathcal{W}_{D}\) is quadratic with respect to \(W_{Y|X}\). Consequently, the set \(\mathcal{W}_{D}\) may be empty for certain values of \(D\), and in general it does not constitute a convex set. The convexity of \(\mathcal{W}_{D}\) follows if the function $\Lambda$ defined below is convex (level sets of a convex function are convex): 
\begin{align}
    \Lambda: W_{Y|X} \longmapsto \sum_{\substack{x,y,z,x' \in \\ \mathcal{X} \times \mathcal{Y} \times \mathcal{Y} \times \mathcal{X}}}
    &d(x,y,z)\; p_X(x)\,p_X(x')\times\nonumber\\ &W_{Y|X}(y|x)\,W_{Y|X}(z|x').\label{conv:eq}
\end{align}

Observe that the convexity of $\Lambda$ depends entirely on the source distribution $p_X$ and the distortion function $d(\cdot,\cdot,\cdot)$. Moreover, since $I(X;Y)$ is a convex function of $W_{Y|X}$ for a fixed source distribution $p_X$, if $\Lambda$ (and equivalently the set $\mathcal{W}_{D}$) is convex, then the minimization in \eqref{Racc:def} becomes a convex optimization problem.
\end{remark}

\begin{definition}
We define $D_{\text{min}}$ and $D_{\text{max}}$ by
\begin{align}
    D_{\text{min}} &\triangleq \min_{W_{Y|X}} \mathbb{E}_{Q_{X,Y,\hat{Y}}}\bigl[d(X,Y,\hat{Y})\bigr], \label{dist:QQ}\\
    D_{\text{max}} &\triangleq \min_{r_Y}  \mathbb{E}_{K_{X,Y,\hat{Y}}}\bigl[d(X,Y,\hat{Y})\bigr], \label{dist:K}
\end{align}
where $r_Y(y)$ is selected from the set of valid probability distributions on $Y$, and
\begin{align}
    Q_{X,Y,\hat{Y}}(x,y,\hat{y}) &= p_X(x)W_{Y|X}(y| x)p_{Y}(\hat{y}),\\
    K_{X,Y,\hat{Y}}(x,y,\hat{y}) &= p_X(x)r_{Y}(y) r_{Y}(\hat{y}).
\end{align}
Here, $p_Y$ is given in \eqref{def:p:Y:acc}.
\end{definition}
Note that $D_{\text{min}} \leq D_{\text{max}}$, and the optimization problem in \eqref{Racc:def} is infeasible for $D < D_{\text{min}}$, while for $D \geq D_{\text{max}}$ the value of $\mathtt{R}_{\mathrm{I}_2}(D)$ is zero (with the choice $W_{Y|X}(y|x)=r_Y(y)$). Thus, for $D \geq D_{\text{max}}$ we have $R(D) = \mathtt{R}_{\mathrm{I}_2}(D) = 0$. 

The quantities $D_{\text{min}}$ and $D_{\text{max}}$ depend entirely on the source distribution $p_X$ and the distortion function $d(\cdot,\cdot,\cdot)$.\,
For certain special choices of $p_X$ and $d(\cdot,\cdot,\cdot)$, it may happen that $D_{\text{min}} = D_{\text{max}}$, in which case $\mathtt{R}_{\mathrm{I}_2}(D)$ is equal to zero for all $D \geq D_{\text{min}}$. Moreover, it is immediate that $\mathtt{R}_{\mathrm{I}_2}(D)$ is a non-increasing function of $D$ for $D \geq D_{\text{min}}$. 
\begin{proposition}\label{racc:con:lem}
Let $\mathcal{X} = \mathcal{Y} = \{0,1\}$ and let $p_X$ be the uniform distribution on $\mathcal{X}$. 
Define
\begin{align}
    e_1 &\triangleq d(0,0,0) + d(0,1,1) - d(0,1,0) - d(0,0,1),\label{e1}\\
    e_2 &\triangleq d(1,0,0) + d(1,1,1) - d(1,1,0) - d(1,0,1).\label{e2}
\end{align}
Then, $\Lambda$ is a   convex function if and only if
\begin{align}
    e_1 = e_2 \geq 0.
\end{align}
\end{proposition}
\begin{proof}
    The proof is found in Appendix \ref{apx:prop-1-proof}.
\end{proof}
\begin{example}
   Let $\mathcal{X} = \mathcal{Y} = \{0,1\}$ and let $p_X$ be the uniform distribution on $\mathcal{X}$. 
Let the distortion function be
\[
d(X,Y,\hat{Y}) = \mathds{1}(X\neq Y) + \gamma\, \mathds{1}(X\neq \hat{Y}),
\]
where $\mathds{1}$ is the indicator function and $\gamma \in [0,1]$ is a constant. 
For this choice of distortion function, it is easy to see that $\Lambda$ is an affine function and that
\[
D_{\text{min}} = \frac{\gamma}{2},
\qquad
D_{\text{max}} = \frac{\gamma+1}{2}.
\] 
\end{example}

The definition of $\mathtt{R}_{\mathrm{I}_2}(D)$ can be extended to Gaussian sources.

\begin{proposition}\label{ex:gua}
Let $X \sim \mathcal{N}(0,\sigma^2)$ and let the distortion function be
\[
d(x,y,\hat{y}) = (x-y)^2 + \gamma (x-\hat{y})^2,
\]
where $\gamma$ is a positive constant. 
%
Then 
the problem of finding $\mathtt{R}_{\mathrm{I}_2}(D)$ is infeasible for $D < \frac{\gamma^2+2\gamma}{1+\gamma}\,\sigma^2$. Moreover, for $D \geq \frac{\gamma^2+2\gamma}{1+\gamma}\,\sigma^2$,
\begin{align}
    \mathtt{R}_{\mathrm{I}_2}(D)=
    \begin{cases}
\frac{1}{2}\log_{2}\left(\frac{(1+\gamma)^{-1}\sigma^2}{D-D_{\text{min}}}\right)     & \text{if } D_{\text{min}}<D< \sigma^2(1+\gamma),\\
0    & \text{if } D\ge \sigma^2(1+\gamma).
\end{cases}
\end{align}
\end{proposition}
\begin{proof}
    The proof is found in Appendix \ref{apx:prop-2-proof}.
\end{proof}

%
The core of the proof lies in demonstrating the optimality of the Gaussian channel $W_{Y|X}$ when evaluating $\mathtt{R}_{\mathrm{I}_2}(D)$.

\subsection{Convexification over $D$}
Since \(R(D)\) is convex in \(D\) (see Remark~\ref{rmk1dd}), and $R(D)\leq \mathtt{R}_{\mathtt{I}_1}(D)$, the largest convex function below $\mathtt{R}_{\mathtt{I}_1}(D)$ is still an upper bound on $R(D)$, and may be strictly smaller than $\mathtt{R}_{\mathtt{I}_1}(D)$. A similar argument can be made for $\mathtt{R}_{\mathtt{I}_2}(D)$.
\begin{definition}
    The lower convex envelope of $\mathtt{R}_{\mathtt{I}_1}(D)$ is denoted $\underline{\mathtt{R}}_{\mathtt{I}_1}(D)$ and defined as the largest convex function that is everywhere less than or equal to $\mathtt{R}_{\mathtt{I}_1}(D)$; mathematically, 
    \begin{align}
        \label{eqn:lower-convex-envelope}
        \underline{\mathtt{R}}_{\mathtt{I}_1}(D)
       & := \sup \bigl\{ g(D) : g \text{ is convex and }\nonumber\\&~~g(y) \le \mathtt{R}_{\mathtt{I}_1}(y)\ \forall\, y \ge 0 \bigr\},
       ~ \forall D \geq D_{\text{min}} .
    \end{align}
    The lower convex envelope $\underline{\mathtt{R}}_{I_2}(D)$ of $\mathtt{R}_{I_2}(D)$ is defined similarly, substituting these quantities at the respective places in (\ref{eqn:lower-convex-envelope}).
\end{definition}

With this definition, we have
$$\underline{\mathtt{R}}_{\mathtt{I}_2}(D)\geq \underline{\mathtt{R}}_{\mathtt{I}_1}(D)\geq R(D) .$$

\begin{proposition}
    If the function $\Lambda$ in \eqref{conv:eq} is convex in $W_{Y|X}$, then $\mathtt{R}_{\mathtt{I}_2}(D)$ is convex in $D$ and hence $\underline{\mathtt{R}}_{\mathtt{I}_2}(D)=\mathtt{R}_{\mathtt{I}_2}(D)$.
\end{proposition}
%

\begin{proof}The proof is found in Appendix \ref{apx:prop-3-proof}.\end{proof}

\section{Outer bounds}
\subsection{Infeasibility via convex envelope}

\begin{theorem}\label{con:acc:th}
    Define
    \begin{align}
        \mathtt{d} = \max_{x,y,z,t \in \mathcal{X} \times \mathcal{Y} \times \mathcal{Y} \times \mathcal{Y}}
        \bigl[ d(x,y,z) - d(x,y,t) \bigr]. \label{new:d}
    \end{align}
    Then, for any coding scheme $P_{Y^n|X^n}$ satisfying
    \begin{align}
        \limsup_{n\to\infty} \mathbb{E}\!\left[ \frac{1}{n}\sum_{i=1}^{n} d(X_i,Y_i,Y_{i-1}) \right] \le D, \label{conv:ass}
    \end{align}
    the achievable rate $R$ for $D \geq D_{\text{min}} -\mathtt{d} $ must obey
    \begin{align}
        R \ge \underline{\mathtt{R}}_{\mathtt{I}_2}(D + \mathtt{d}).
    \end{align}
    Thus for $D \geq D_{\text{min}} -\mathtt{d} $, we have 
    \begin{align}
\underline{\mathtt{R}}_{\mathtt{I}_2}(D)\geq R(D)\geq \underline{\mathtt{R}}_{\mathtt{I}_2}(D + \mathtt{d}). \label{lower:ach:acc}
    \end{align}
\end{theorem}
\begin{proof}
    The proof is found in Appendix \ref{apx:thm-3-proof}.
\end{proof}

Note that, when $D \geq D_{\text{max}} - \mathtt{d}$, the lower bound in \eqref{lower:ach:acc} becomes trivial.
\begin{remark}
Note that, when $d(x,y,\hat y)$ depends only on $(x,y)$, we have $\mathtt{d} = 0$ and the lower and upper bounds of $R(D)$ coincide. This recovers the classical rate-distortion bound where the distortion depends on $X_i$ and $Y_i$, and not the past symbol $Y_{i-1}$. 
\end{remark}
\subsection{Infeasibility via relaxation}
We relax the original problem by allowing the decoder $\mathcal{D}$ to map each message $m \in \mathcal{M} = \{1,2,\dots,2^{\lfloor nR \rfloor}\}$ to two reconstruction sequences $y^n, \hat{y}^n \in \mathcal{Y}^n$. 
The average distortion is measured as
\begin{equation}
  \bar{\mathsf{d}}(x^n, y^n, \hat{y}^n) \triangleq \frac{1}{n} \sum_{i=1}^{n} d(x_i, y_i, \hat{y}_i).
\end{equation}
If we restrict the decoder to choose $\hat{y}_i = y_{i-1}$, this reduces to the original problem formulation, whereas here we permit the decoder to select the sequences $y^n$ and $\hat{y}^n$ independently. 
The solution of this relaxed problem yields a lower bound on the original problem. 

Observe that the relaxed problem is a standard rate–distortion problem with reproduction alphabet $\tilde{\mathcal{Y}} = \mathcal{Y} \times \mathcal{Y}$ and reproduction symbol $\tilde{Y} = (Y,\hat{Y})$, where the distortion at time $i$ depends only on the current symbols $(x_i,\tilde{y}_i) = (x_i,y_i,\hat{y}_i)$. 
Hence any achievable rate $R$ must satisfy
\[
R \;\geq\; \mathtt{R}_{O_1}(D),
\]
where
\begin{align}
    \mathtt{R}_{O_1}(D)
    &= \min_{W_{Y,\hat{Y}|X} \in \mathcal{V}_D} I(X;Y,\hat{Y}), \label{reach:prob}
\end{align}
and
\begin{align}
    \mathcal{V}_D
    &\triangleq \Bigl\{ W_{Y,\hat{Y}|X} : \mathbb{E}_{p_X(x) W_{Y,\hat{Y}|X}(y,\hat{y}\mid x)}\bigl[d(X,Y,\hat{Y})\bigr] \leq D \Bigr\}. \label{def:V}
\end{align}

\begin{remark}
In the above construction, we introduce $\hat{y}^n$ as a new variable, which allows us to replace $d(x_i,y_i,y_{i-1})$ with $d(x_i,y_i,\hat{y}_i)$, eliminating the memory for each $i$, at the cost of relaxing the problem and obtaining an outer bound. It is also possible to just remove the memory, say in even time instances. For instance, if we assume that the decoder is restricted to choose $\hat{y}_{i} = y_{i-1}$ for even $i$, but has no restriction on $\hat{y}_{i}$ for odd $i$, one can obtain better infeasibility results: 
\begin{align}
    \mathtt{R}_{O_2}(D) = \min \frac{1}{2} I(X_1,X_2; Y_1,Y_2,\hat{Y}_1),
\end{align}
where the minimum is taken over all kernels $W_{Y_1,Y_2,\hat{Y}_1|X_1,X_2}$ such that under the joint distribution $p_{X}(x_1)p_X(x_2)W_{Y_1,Y_2,\hat{Y}_1|X_1,X_2}(y_1,y_2,\hat{y}_1|x_1,x_2)$, we have
\begin{equation*}
    \mathbb{E}\bigl[d(X_1,Y_1,\hat{Y}_1) + d(X_2,Y_2,Y_1)\bigr] \leq 2D.
\end{equation*}
\end{remark}

\section{Numerical Examples}
\label{sec:Examples}

Consider a problem with binary $\{0,1\}$ source symbols and representations, and Hamming distortion, i.e. unit distortion for $x_i \neq y_i$, and zero distortion otherwise. However, the decoder outputs $y_i$ have the property that $(y_{i-1},y_i) = (0,1)$ has a cost of $c$, perhaps because generating a string of $1$ has a one-time startup cost. Then the subjective distortion function can be represented using two matrices, as follows:
\begin{figure}
    \centering
    \includegraphics[width=0.9\columnwidth]{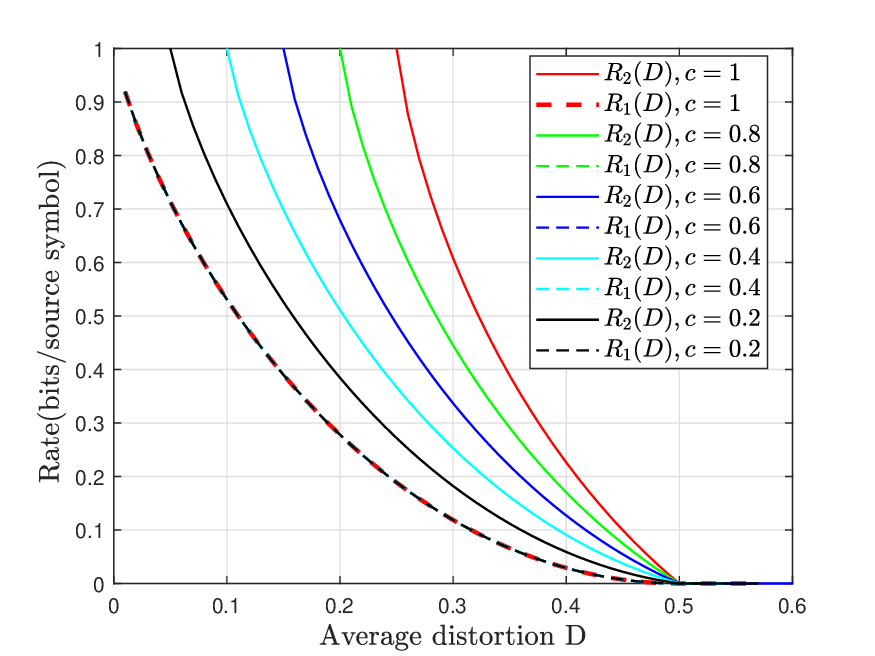}
    \caption{Example of outer bound $R_{1}(D)$ and inner bound $R_{2}(D)$ for the distortion function given in (\ref{eqn:distortion-example-1})-(\ref{eqn:distortion-example-2}), with various values of the cost $c$. A uniform prior is applied to the source. Note that the outer bound is invariant to $c$.}
    \label{fig:Figure2}
\end{figure}
\begin{align}
    \label{eqn:distortion-example-1}
    d(x_i,y_i,y_{i-1} = 0) =  & \left[\begin{array}{cc} 0 & 1+c \\ 1 & c \end{array}\right],\\ 
    \label{eqn:distortion-example-2}
    d(x_i,y_i,y_{i-1} = 1) =  & \left[\begin{array}{cc} 0 & 1 \\ 1 & 0 \end{array}\right],
\end{align}
with $x_i$ indexed to the rows and $y_i$ to the columns. The above example is a simplified model of {\em carbon catabolite repression} \cite{aidelberg2014hierarchy}, in which bacteria must guess which sugar will be available in a nutrient supply, and express the appropriate gene to metabolize it; making a mistake is costly, but so is starting a genetic circuit to metabolize a new sugar. Obviously this system cannot query a genie, along the lines of our setup in Section \ref{sec:setup}, but our goal is to establish absolute performance limits for this type of problem. Future work will consider how a simplified encoder and decoder may achieve the limits, such as by using a single-letter code.

For convenience in labeling the plots, define:
\begin{align}
    R_1(D) &\triangleq \max\Big\{\mathtt{R}_{O_1}(D),\mathtt{R}_{\mathtt{I}_2}(D+\mathtt{d}),0\Big\} , \\
    R_2(D) &\triangleq \mathtt{R}_{\mathtt{I}_1}(D) ,
\end{align}
so that $R_1(D)$ is the best of the outer bounds (excluding $O_2$ or higher-order relaxation), and $R_2(D)$ is the best of the inner bounds.
%
%
Then, for any  distortion $D$ we must have
\begin{align}
    R_{1}(D)\leq R(D)\leq R_{2}(D).
\end{align}
%
For this example, the outer and inner bounds ($R_{1}(D)$ and $R_{2}(D)$) are plotted in Fig.~\ref{fig:Figure2} for various values of $c$, assuming a uniform prior on the source $X$. 
In this setting, we observe that $R_{1}(D) = \mathtt{R}_{O_1}(D)$ and that $\mathtt{R}_{O_1}(D)$ does not depend on $c$ for all positive $c$.
To highlight this invariance, the curve corresponding to $c=1$ is drawn with a thicker line, and all $R_{1}(D)$ curves coincide over the entire distortion range.

\color{black}

\appendix

\subsection{Expansion of Remark 1}
\label{apx:Remark-1}

Here we show that $R(D)$ is convex. Following (\ref{eqn:limsup}) and the discussion after the equation, let
\begin{align}
    \label{eqn:RD-definition}
    \nonumber\lefteqn{R(D) =}&\\ 
    &\inf \left\{R : \limsup_{n \to \infty} \mathbb{E}\Bigl[ \bar{\mathsf{d}}\bigl( X^n, \mathcal{D}_n^{(R)} \circ \mathcal{E}_n^{(R)}(X^n) \bigr) \Bigr] \leq D \right\} ,
\end{align}
where we write the decoder and encoder as $\mathcal{D}_n^{(R)}$ and $\mathcal{E}_n^{(R)}$, respectively, to emphasize that they operate at rate $R$. Consider $D_1$ and $D_2$, with respective rates $R(D_1)$ and $R(D_2)$.

Consider the following time-sharing scheme:
\begin{enumerate}
    \item Let $X$ be a source, and let $X_1,X_2,\ldots,X_n$ be $n$ letters from the source. Let $0 \leq \lambda \leq 1$ be a constant, and let $\hat{\lambda} = \frac{\lfloor \lambda n \rfloor}{n}$.
    
    \item Divide the source into two parts, the first part consisting of letters $X_1,\ldots,X_{\hat{\lambda}n}$, and the second part consisting of $X_{\hat{\lambda}n +1},\ldots,X_n$. Note that the first part contains $\hat{\lambda}n$ letters, and the second part contains $(1-\hat{\lambda})n$ letters.
    
    \item Let $(\mathcal{E}_n^{(R(D_1))},\mathcal{D}_n^{(R(D_1))})$ and $(\mathcal{E}_n^{(R(D_2))},\mathcal{D}_n^{(R(D_2))})$ be the encoder-decoder pairs satisfying (\ref{eqn:RD-definition}) for $D_1$ and $D_2$, respectively. For the two parts of the source in step 2, encode the first part with $(\mathcal{E}_n^{(R(D_1))},\mathcal{D}_n^{(R(D_1))})$ and the second part with $(\mathcal{E}_n^{(R(D_2))},\mathcal{D}_n^{(R(D_2))})$.
\end{enumerate}
The rate of this scheme is given by 
\begin{align}
    \bar{R} \triangleq R(D_1) \hat{\lambda} + R(D_2) (1-\hat{\lambda})
\end{align}
bits per source symbol, while the expected distortion can be written (see (\ref{eqn:average-distortion})-(\ref{eqn:limsup})):
\begin{align}
    \nonumber\lefteqn{\mathbb{E}\Big[\bar{\mathsf{d}}(X^n,Y^n)\Big]} & \\
    &= 
    \frac{1}{n} \mathbb{E}\left[\sum_{i=1}^{\hat{\lambda} n} d(X_i, Y_i, Y_{i-1}) + \sum_{i=\hat{\lambda} n + 1}^{n} d(X_i, Y_i, Y_{i-1})\right] \nonumber\\
    &= 
    \frac{1}{n} \mathbb{E}^{(1)}\left[\sum_{i=1}^{\hat{\lambda} n} d(X_i, Y_i, Y_{i-1})\right] \nonumber\\
    &+\frac{1}{n} \mathbb{E}\big[d(X_{\hat{\lambda} n + 1}, Y_{\hat{\lambda} n + 1}, Y_{\hat{\lambda} n})-d(X_{\hat{\lambda} n + 1}, Y_{\hat{\lambda} n + 1}, y_{0})\big]\nonumber\\&+\mathbb{E}^{(2)}\left[d(X_{\hat{\lambda} n + 1}, Y_{\hat{\lambda} n + 1}, y_{0})+\sum_{i=\hat{\lambda} n + 2}^{n} d(X_i, Y_i, Y_{i-1}) \right]\nonumber\\
    \label{eqn:bar-D}
    &\leq \hat{\lambda} D_1 + (1-\hat{\lambda}) D_2+ 
    \\
    &+\frac{1}{n} \mathbb{E}\big[d(X_{\hat{\lambda} n + 1}, Y_{\hat{\lambda} n + 1}, Y_{\hat{\lambda} n})-d(X_{\hat{\lambda} n + 1}, Y_{\hat{\lambda} n + 1}, y_{0})\big]\nonumber,
\end{align}
where $y_0$ is the fixed symbol used in the definition of \eqref{eqn:average-distortion}, 
$\mathbb{E}^{(1)}$ takes the expectation with respect to $(\mathcal{E}_n^{(R(D_1))},\mathcal{D}_n^{(R(D_1))})$, and $\mathbb{E}^{(2)}$ takes it with respect to $(\mathcal{E}_n^{(R(D_2))},\mathcal{D}_n^{(R(D_2))})$, and where (\ref{eqn:bar-D}) is an application of Definition \ref{defn:1}. Define $\bar{D} \triangleq D_1 \hat{\lambda} + D_2 (1-\hat{\lambda})$. Since the distortion function is bounded, we have
$$\lim_{n\rightarrow\infty}\frac{1}{n} \mathbb{E}\big[d(X_{\hat{\lambda} n + 1}, Y_{\hat{\lambda} n + 1}, Y_{\hat{\lambda} n})-d(X_{\hat{\lambda} n + 1}, Y_{\hat{\lambda} n + 1}, y_{0})\big]=0.$$
Thus, the pair $(\bar{R},\bar{D})$ is achievable by time-sharing, and is included in the set under the $\limsup$ in (\ref{eqn:RD-definition}). Thus, $R(\bar{D}) \leq \bar{R}$, which is sufficient to show that $R(D)$ is convex.

\subsection{Proof of Proposition 1}
\label{apx:prop-1-proof}

    We parameterize kernel  $W_{Y|X}$ as the matrix
    \begin{align}
           W_{Y|X} = \begin{pmatrix}
        \alpha & 1-\alpha \\[2pt]
        \beta  & 1-\beta
    \end{pmatrix},
    \qquad \alpha,\beta \in [0,1].\label{2:v:ker}
    \end{align}
 
    We examine the Hessian of function $\Lambda$
    with respect to $(\alpha,\beta)$.  Computing second derivatives gives
    \begin{align}
        \operatorname{Hess}(\Lambda) =\frac{1}{2} \begin{pmatrix}
            2e_1 & e_1+e_2 \\[4pt]
            e_1+e_2 & 2e_2
        \end{pmatrix},
    \end{align}
    where
    \begin{align*}
        e_1 &= d(0,0,0) + d(0,1,1) - d(0,1,0) - d(0,0,1), \\
        e_2 &= d(1,0,0) + d(1,1,1) - d(1,1,0) - d(1,0,1).
    \end{align*}
    For $\operatorname{Hess}(\Lambda)$ to be positive semidefinite we require $e_1 > 0$, $e_2 > 0$, and $\det\bigl(\operatorname{Hess}(\Lambda)\bigr) \ge 0$.  
    A direct computation yields
    \[
        \det\bigl(\operatorname{Hess}(\Lambda)\bigr) = -\frac{1}{4}(e_1 - e_2)^2,
    \]
    which is non‑negative only when $e_1 = e_2$.  
    Hence $\lambda$ is  non-linear convex precisely when $e_1 = e_2 > 0$, 
    which completes the proof.

\subsection{Proof of Proposition 2}
\label{apx:prop-2-proof}

We first prove that $D_{\text{min}} = \frac{\gamma^2+2\gamma}{1+\gamma}\,\sigma^2$. 
We wish to minimize
\[
\mathbb{I}\triangleq \mathbb{E}\bigl[(X-Y)^2\bigr] + \gamma\, \mathbb{E}\bigl[(X-\hat{Y})^2\bigr]
\]
over all conditional distributions $W_{Y|X}$. 
Since $\hat{Y}$ is independent of $(X,Y)$ and has the same distribution as $Y$, the quantity $\mathbb{I}$ simplifies to
\[
\mathbb{I} = (1+\gamma)\sigma^2 + (1+\gamma)\mathbb{E}[Y^2] - 2\,\mathbb{E}[XY].
\]
Let $\mathbb{E}[Y^2]=b^2$. By the Cauchy–Schwarz inequality,
\[
\mathbb{E}[XY] \leq \sqrt{\mathbb{E}[X^2]\,\mathbb{E}[Y^2]} = \sigma b,
\]
which yields
\begin{align}
   \mathbb{I} 
   &\geq (1+\gamma)\sigma^2 + (1+\gamma)b^2 - 2\sigma b. \nonumber
\end{align}
The minimum of $(1+\gamma)\sigma^2 + (1+\gamma)b^2 - 2\sigma b$ over $b$ is equal to $\dfrac{\gamma^2+2\gamma}{1+\gamma}\sigma^2$. 
Thus, for every conditional distribution $W_{Y|X}$ we have
\[
\mathbb{I} \geq \frac{\gamma^2+2\gamma}{1+\gamma}\,\sigma^2.
\]
Moreover, by choosing $Y = \sqrt{\dfrac{\gamma^2+2\gamma}{1+\gamma}}\,X$ we achieve equality. Thus, $D_{\text{min}} = \frac{\gamma^2+2\gamma}{1+\gamma}\,\sigma^2$, which implies $\mathtt{R}_{\mathrm{I}_2}(D)$ is infeasible for $D < \frac{\gamma^2+2\gamma}{1+\gamma}\,\sigma^2$, thus proving the first part of the proposition.

To prove the second part of the proposition,
the constraint $\mathbb{E}\bigl[d(X,Y,\hat{Y})\bigr] \leq D$ simplifies to
\begin{align}
    (1+\gamma)\mathbb{E}[Y^2] - 2\mathbb{E}[XY] + (1+\gamma)\sigma^2 - D \leq 0. \label{condd::1}
\end{align}
Assume $D < (1+\gamma)\sigma^2$. 
For any conditional distribution $p_{Y|X}$ satisfying \eqref{condd::1}, we have
\begin{align}
I(X;Y)
&= h(X) - h(X|Y)\nonumber
  \\&= h(X) - h(X-\alpha Y \mid Y) \nonumber\\
&\overset{(a)}{\geq} h(X) - h(X-\alpha Y)\nonumber
 \\&\overset{(b)}{\geq} \frac{1}{2}\log\!\left(\frac{\sigma^2}{\mathbb{E}\bigl[(X-\alpha Y)^2\bigr]}\right), \label{ineq::1:1}
\end{align}
where $h(\cdot)$ denotes differential entropy. 
Step (a) uses the fact that conditioning reduces differential entropy, and step (b) follows because, for a fixed second moment, the Gaussian distribution maximizes differential entropy together with the expression of the differential entropy of a Gaussian random variable. 
Since \eqref{ineq::1:1} holds for every $\alpha \in \mathbb{R}$, it also holds when we maximize the right-hand side over $\alpha$, i.e.,
\begin{align}
I(X;Y)
&\geq \max_{\alpha} \frac{1}{2}\log\!\left(\frac{\sigma^2}{\mathbb{E}\bigl[(X-\alpha Y)^2\bigr]}\right). \label{ineq::1:2}
\end{align}
It is immediate that the maximizer satisfies
\[
\alpha^{\star} = \frac{\mathbb{E}[XY]}{\mathbb{E}[Y^2]},
\]
so substituting $\alpha^{\star}$ into \eqref{ineq::1:2} yields
\begin{align}
I(X;Y)
&\geq \frac{1}{2}\log\!\left(\frac{\sigma^2}{\sigma^2 - \frac{\mathbb{E}^2[XY]}{\mathbb{E}[Y^2]}}\right). \label{ineq::1:3}
\end{align}
Let
\[
\widetilde{D} \triangleq (1+\gamma)\sigma^2 - D,\qquad
a \triangleq \mathbb{E}[Y^2],\qquad
b \triangleq \mathbb{E}[XY].
\]
Then \eqref{ineq::1:3} can be rewritten as
\begin{align}
I(X;Y)
&\geq \frac{1}{2}\log\!\left(\frac{\sigma^2}{\sigma^2 - \frac{b^2}{a}}\right). \label{ineq::1:4}
\end{align}
Moreover, the constraint \eqref{condd::1} becomes
\[
(1+\gamma)a - 2b + \widetilde{D} \leq 0.
\]
Since \eqref{ineq::1:4} holds for every pair $(a,b)$ satisfying $(1+\gamma)a - 2b + \widetilde{D} \leq 0$, we obtain
\begin{align}
I(X;Y)
&\geq \min_{a,b :\; (1+\gamma)a - 2b + \widetilde{D} \leq 0}
     \frac{1}{2}\log\!\left(\frac{\sigma^2}{\sigma^2 - \frac{b^2}{a}}\right). \label{ineq::1:5}
\end{align}
A direct calculation shows that the minimizers of \eqref{ineq::1:5} are
\[
a^{\star} = \frac{\widetilde{D}}{1+\gamma},
\qquad
b^{\star} = \widetilde{D}.
\]
Substituting $(a^{\star}, b^{\star})$ into \eqref{ineq::1:5} yields
\begin{align}
    I(X;Y)
    &\geq \frac{1}{2}\log_{2}\!\left(\frac{(1+\gamma)^{-1}\sigma^2}{D - D_{\text{min}}}\right), \label{ineq::1:6}
\end{align}
where
\[
D_{\text{min}} \triangleq \frac{\gamma^2+2\gamma}{1+\gamma}\,\sigma^2.
\] 
Now consider $(X,Y)$ jointly Gaussian with covariance matrix
\[
\begin{pmatrix}
        \sigma^2 & \widetilde{D} \\[2pt]
        \widetilde{D} & \dfrac{\widetilde{D}}{1+\gamma}
\end{pmatrix}.
\]
A straightforward computation shows that equality holds in \eqref{ineq::1:6} for this choice, which completes the proof. 

\subsection{Proof of Proposition 3}
\label{apx:prop-3-proof}

Let $D_1$ and $D_2$ be distortions with respective inner bounds $\mathtt{R}_{\mathtt{I}_2}(D_1)$ and $\mathtt{R}_{\mathtt{I}_2}(D_2)$. For any $\lambda \in [0,1]$, and letting $\bar{D} \triangleq \lambda D_1 + (1-\lambda)D_2$, we require that 
\begin{align}
    \lambda \mathtt{R}_{\mathtt{I}_2}(D_1) + (1-\lambda) \mathtt{R}_{\mathtt{I}_2}(D_2) &\geq \mathtt{R}_{\mathtt{I}_2}(\bar{D}).
\end{align}

Let $W_{1,Y|X} \triangleq \arg \min_{W_{Y|X} \in \mathcal{W}_{D_1}} I(X;Y)$ be a solution for $\mathtt{R}_{\mathtt{I}_2}(D_1)$ in (\ref{Racc:def}), and similarly $W_{2,Y|X}$ for $\mathtt{R}_{\mathtt{I}_2}(D_2)$.
Define 
\begin{align}
    \bar{W}_{Y|X} &= \lambda W_{1,Y|X} + (1-\lambda) W_{2,Y|X} .
\end{align}
Using $\bar{W}_{Y|X}$, the average distortion is
\begin{align}
    \mathbb{E}_{\bar{W}_{Y|X}}\Big[d(X,Y,\hat{Y})\Big] 
    \label{eqn:prop-3-convex-1}
    &=
    \Lambda(\bar{W}_{Y|X}) \\
    \label{eqn:prop-3-convex-2}
    &\leq \lambda \Lambda(W_{1,Y|X}) + (1-\lambda) \Lambda(W_{2,Y|X}) \\
    \label{eqn:prop-3-convex-3}
    &\leq \lambda D_1 + (1-\lambda)D_2 \\
    &= \bar{D},
\end{align}
where: (\ref{eqn:prop-3-convex-1}) follows from the definition of $\Lambda$ in (\ref{conv:eq}); (\ref{eqn:prop-3-convex-2}) follows because $\Lambda$ is convex, by assumption; and (\ref{eqn:prop-3-convex-3}) follows from (\ref{def:W}) and the definitions of $D_1$ and $D_2$. Because $\Lambda(\bar{W}_{Y|X}) \leq \bar{D}$, it follows that $\bar{W}_{Y|X} \in \mathcal{W}_{\bar{D}}$. Thus, 
\begin{align}
    I_{p_X\bar{W}_{Y|X}}(X;Y) &\geq \min_{W_{Y|X} \in \mathcal{W}_{\bar{D}}} I(X;Y) \\
    &= \mathtt{R}_{\mathtt{I}_2}(\bar{D}).\label{eqnUb1s}
\end{align}
On the other hand, since $I(X;Y)$ is convex in $p_{Y|X}$ for a fixed input distribution $p_X$, we have
\begin{align}\nonumber
I_{p_X\bar{W}_{Y|X}}(X;Y)&\leq \lambda I_{p_X{W}_{1,Y|X}}(X;Y)\\&\nonumber\qquad+(1-\lambda)I_{p_X{W}_{2,Y|X}}(X;Y)
\\&=\lambda \mathtt{R}_{\mathtt{I}_2}(D_1) + (1-\lambda) \mathtt{R}_{\mathtt{I}_2}(D_2).\label{eqn:mutual-info-bound}
\end{align}
Finally, from \eqref{eqnUb1s} and (\ref{eqn:mutual-info-bound}), we deduce that
\begin{align*}
    \mathtt{R}_{\mathtt{I}_2}(\bar{D})\leq \lambda \mathtt{R}_{\mathtt{I}_2}(D_1) + (1-\lambda) \mathtt{R}_{\mathtt{I}_2}(D_2),
\end{align*}
establishing the convexity of $\mathtt{R}_{\mathtt{I}_2}(D)$ in $D$. To complete the proof, since $\underline{\mathtt{R}}_{\mathtt{I}_2}(D)$ is the largest convex function less than or equal to $\mathtt{R}_{\mathtt{I}_2}(D)$, this implies $\underline{\mathtt{R}}_{\mathtt{I}_2}(D) = \mathtt{R}_{\mathtt{I}_2}(D)$, which proves the proposition.

\subsection{Proof of Theorem 3}
\label{apx:thm-3-proof}

    The Markov chain $X^n \to M \to Y^n$ holds. Let $D\geq D_{1,\text{min}}-\mathtt{d}$. We have
    \begin{align}
        nR &\ge H(M) \overset{(a)}{\ge} I(X^n;Y^n) 
             = H(X^n) - H(X^n|Y^n)\nonumber \\
           &\overset{(b)}{=} \sum_{i=1}^{n} H(X_i) - H(X_i|Y^n, X^{i-1})\nonumber \\
           &\overset{(c)}{\ge} \sum_{i=1}^{n} H(X_i) - H(X_i|Y_i)
            = \sum_{i=1}^{n} I(X_i;Y_i),\label{acc:conv:1:1}
    \end{align}
    where (a) follows from the data‑processing inequality, (b) uses the i.i.d. property of $X^n$ together with the chain rule, and (c) holds because conditioning reduces entropy.  

   Let $p_{X_i,Y_i}$ denote the joint distribution of $(X_i,Y_i)$.  
    Independently generate $\hat{Y}_{i-1} \sim p_{Y_i}$ (i.e., from the same marginal as $Y_i$).  
    By the definition of $\mathtt{R}_{\mathrm{I}_2}(D)$ (see (\ref{Racc:def})):
    \[
        I(X_i;Y_i) \ge \mathtt{R}_{\mathrm{I}_2}\!\bigl( \mathbb{E}[d(X_i,Y_i,\hat{Y}_{i-1})] \bigr).
    \]
    Moreover,
    \begin{align}
        \mathbb{E}[d(X_i,Y_i,\hat{Y}_{i-1})]
        &= \mathbb{E}[d(X_i,Y_i,Y_{i-1})] 
           \nonumber\\&+ \mathbb{E}[d(X_i,Y_i,\hat{Y}_{i-1}) - d(X_i,Y_i,Y_{i-1})]\nonumber \\
        &\overset{(a)}{\le} \mathbb{E}[d(X_i,Y_i,Y_{i-1})] + \mathtt{d},
    \end{align}
where $(a)$ follows from the definition of $\mathtt{d}$ in \eqref{new:d}.
Since $\mathtt{R}_{\mathrm{I}_2}(\cdot)$ and its lower convex envelope $\underline{\mathtt{R}}_{\mathrm{I}_2}(\cdot)$ are non-increasing functions, from \eqref{acc:conv:1:1} we obtain
    \begin{align}
        nR &\ge \sum_{i=1}^{n} I(X_i;Y_i) 
             \ge \sum_{i=1}^{n} \mathtt{R}_{\mathrm{I}_2}\!\bigl( \mathbb{E}[d(X_i,Y_i,\hat{Y}_{i-1})] \bigr)\nonumber \\
           &\overset{(a)}{\ge} \sum_{i=1}^{n} \underline{\mathtt{R}}_{\mathrm{I}_2}\!\bigl( \mathtt{d} + \mathbb{E}[d(X_i,Y_i,Y_{i-1})] \bigr) \nonumber\\
           &= n \; \frac{1}{n} \sum_{i=1}^{n} 
                \underline{\mathtt{R}}_{\mathrm{I}_2}\!\bigl( \mathtt{d} + \mathbb{E}[d(X_i,Y_i,Y_{i-1})] \bigr)\nonumber \\
           &\overset{(b)}{\ge} n \; 
                \underline{\mathtt{R}}_{\mathrm{I}_2}\!\Bigl( \mathtt{d} + \frac{1}{n} \sum_{i=1}^{n} \mathbb{E}[d(X_i,Y_i,Y_{i-1})] \Bigr), \label{chain:1}
    \end{align}
    where $(a)$ uses $\mathtt{R}_{\mathrm{I}_2} \geq \underline{\mathtt{R}}_{\mathrm{I}_2}$ and $(b)$ follows from Jensen’s inequality applied to the convex function $\underline{\mathtt{R}}_{\mathrm{I}_2}(\cdot)$. 
Cancelling $n$ in \eqref{chain:1} and letting $n \to \infty$, the assumption \eqref{conv:ass}, together with the fact that $\underline{\mathtt{R}}_{\mathrm{I}_2}(\cdot)$ is non-increasing, yields
\[
    R \ge \underline{\mathtt{R}}_{\mathrm{I}_2}(\mathtt{d} + D),
\]
which completes the proof.


\bibliographystyle{ieeetr}
\bibliography{main}

\end{document}